\documentclass[preprint,11pt]{elsarticle}

\usepackage{amssymb}
\usepackage{amsmath}
\usepackage{amsthm}

\usepackage{array}
\usepackage{cases}
\usepackage{helvet}
\usepackage{multirow}
\usepackage{amsfonts}
\usepackage[english]{babel}
\usepackage{graphicx}
\usepackage{url}
\usepackage{bm}
\usepackage{enumitem} 
\usepackage{multirow}
\usepackage{booktabs}
\usepackage{algorithm}
\usepackage{algorithmic}
\usepackage{booktabs}
\usepackage{threeparttable}
\usepackage{graphicx}
\usepackage{subcaption} 
\usepackage{natbib}
\usepackage{epstopdf}
\usepackage{setspace}

\usepackage{geometry}
\numberwithin{equation}{section}

\journal{Journal of Mathematical Analysis and Applications}

\begin{document}

\newtheorem{definition}{$\mathbf{Definition}$}[section]
\newtheorem{theorem}{$\mathbf{Theorem}$}[section]
\newtheorem{corollary}{$\mathbf{Corollary}$}[section]
\newtheorem{lemma}{$\mathbf{Lemma}$}[section]
\newtheorem{proposition}{$\mathbf{Proposition}$}[section]
\newtheorem{property}{$\mathbf{Property}$}[section]
\newtheorem{remark}{$\mathbf{Remark}$}[section]
\newtheorem{example}{$\mathbf{Example}$}[section]

\begin{frontmatter}

\title{A Law of Large Numbers with Convergence Rate based on Nonlinear Expectation Theory and Its Application to Communication Detection}

\author[label1]{Jialiang Fu} \ead{fujialiang@amss.ac.cn}
\author[label1,label2]{Wen-Xuan Lang} \ead{langwx@amss.ac.cn}

\affiliation[label1]{organization={Academy of Mathematics and Systems Science, Chinese Academy of Sciences},
    city={Beijing},
    postcode={100190}, 
    country={China}}
\affiliation[label2]{organization={National Center for Mathematics and Interdisciplinary Sciences, Chinese Academy of Sciences},
    city={Beijing},
    postcode={100190}, 
    country={China}}

\date{} 

\begin{abstract}
In this paper, we establish a new law of large numbers with the rate of convergence for special partial sums in a probability space. The proof relies on nonlinear expectation theory, as the uncertainty of random variables in the special partial sums induces the sublinearity of the expectation. As an application, we apply the new theorem to analyze the feedback channel-based detection problem of non-i.i.d. input signals in communication systems. Specifically, we investigate the convergence rates of the upper probabilities of the detection errors within the sublinear expectation space.
\end{abstract}


\end{frontmatter}



\section{Introduction}\label{sec1}

Nonlinear expectation theory, first introduced by Peng \cite{peng2007G,peng2007}, provides a powerful framework to address uncertainty in probability models, making it a valuable tool in fields such as finance and statistical analysis. A typical nonlinear expectation is the sublinear expectation, under which the expectation is no longer linear but sublinear. Based on the representation theorem, under certain conditions, a sublinear expectation can be expressed as the supremum of a family of linear expectations. By incorporating the distributional uncertainty, this theory offers a more flexible and robust approach for modeling complex systems.

Two of the fundamental results of nonlinear expectation theory are the law of large numbers (LLN) and the central limit theorem (CLT) under sublinear expectations. Peng \cite{Peng19,peng2008new,peng2007law} introduced the concepts of independence and identical distribution for random variables under sublinear expectations, and established a CLT and a weak form of LLN for independent and identically distributed (i.i.d.) random variables in this framework, which demonstrate that the maximal distribution and the $G$-normal distribution are fundamental and widely applicable in nonlinear expectation theory. Chen \cite{Ch16} established a strong form of LLN for i.i.d. sequences under sublinear expectations, investigating the cluster point behavior of the means of the partial sums, and proved a central limit theorem \cite{chen2022a} for the sequence of random variables under the assumption that the probability measure family is rectangular (see \cite{epstein2003recursive} for more). While sample sizes in practice may often be large or seem sufficient for the purposes to handle, the accuracy of the approximations under LLN and CLT can still vary depending on the sample size and other factors. This underscores the importance of assessing the quality of these approximations.

In recent years, there has been considerable literature on estimating convergence rates of limit theorems under sublinear expectations. Song \cite{Song20, Song21} obtained the estimation of convergence rates for Peng's CLT and LLN by Stein's method, where the latter is that for an i.i.d. sequence $\{X_k\}_{k=1}^\infty$ under a sublinear expectation $\mathbb{E}$ with $\underline{\mu}=-\mathbb{E}[-X_1]$, $\overline{\mu}=\mathbb{E}[X_1]$, $\mathbb{E}[|X_1|^2]<\infty$ and for each $\phi\in C_{Lip}(\mathbb{R})$ with Lipschitz constant $L_{\phi}$,
\begin{equation}\label{cr_song}
\bigg|\mathbb{E}\left[\phi\left(  \frac{X_1+\cdots+X_n}{n}\right) \right] -\sup_{y\in[\underline{\mu},\overline{\mu}]}\phi(y)\bigg|\leq L_{\phi}\frac{C}{\sqrt{n}},
\end{equation}
where $C$ is a constant depending only on $\mathbb{E}[|X_1|^2]$. Hu et al. \cite{Hu21} improved the result presented in \cite{Song21}. By the improved result we can know that if $\{X_k\}_{k=1}^\infty$ is an i.i.d. sequence under a sublinear expectation $\mathbb{E}$ with $\underline{\mu}=-\mathbb{E}[-X_1]$, $\overline{\mu}=\mathbb{E}[X_1]$, $C_{\alpha}:=\mathbb{E}[|X_1|^{1+\alpha}]<\infty$ for some $\alpha\in(0,1]$, then for each $\varphi\in C_{Lip}(\mathbb{R})$ with Lipschitz constant $L_{\varphi}$,
\begin{equation}\label{cr_hu}
\left|\mathbb{E}\left[\varphi\left(\frac{X_1+\cdots+X_n}{n}\right)\right]-\max_{r\in[\underline{\mu},\overline{\mu}]}\varphi(r)\right|\leq L_{\varphi}\left(\frac{4C_{\alpha}}{n^\alpha}\right)^{\frac{1}{1+\alpha}}.
\end{equation}
Comparing these two results, (\ref{cr_song}) obtained the $O(1/\sqrt{n})$ convergence rate under $\mathbb{E}[|X_1|^2]<\infty$, while (\ref{cr_hu}) obtained the $O(1/n^{\frac{\alpha}{1+\alpha}})$ convergence rate under the weaker condition, $\mathbb{E}[|X_1|^{1+\alpha}]<\infty$ for some $\alpha\in(0,1]$. Fang et al. \cite{Fang19} referenced the convergence rate from \cite{Song20} to give an approximation and a representation of the $G$-normal distribution, and derived a new central limit theorem with the convergence rate for the suitably normalized partial sums of i.i.d. random variables in a probability space. Recently, there has also been literature on LLNs for non-independent sequences under sublinear expectations, such as $m$-dependent sequences (Gu and Zhang \cite{Gu24}), blockwise $m$-dependent sequences (Fu \cite{Fu25a}), pseudo-independent sequences (Guo and Li \cite{Guo21}, Fu \cite{Fu25b}), and others. Corresponding to the CLT in \cite{Fang19}, in this paper, we utilize the estimation of the convergence rate for Peng's LLN in \cite{Hu21} and combine the technique of measurable approximations from \cite{Fang19}, to establish a new law of large numbers with the convergence rate for means of special partial sums $\frac{(\zeta_1Z_1+K_1)^2\cdots+(\zeta_nZ_n+K_n)^2}{n}$ in a probability space. The sequence $\zeta:=\{\zeta_i\}_{i=1}^\infty$ is uncertain and every $\zeta_i$ lies within the interval $[\underline{\zeta},\overline{\zeta}]$, which results in the sequence $\{(\zeta_iZ_i+K_i)^2\}_{i=1}^\infty$ being non-i.i.d. The proof of our new law of large numbers requires sublinear expectation, where the sublinearity arises from the uncertainty associated with the bounded predictable sequences.

Moreover, to demonstrate the practical significance of our result, we apply the new law of large numbers to the feedback channel-based detection problem of non-i.i.d. input signals in communication systems. The law of large numbers proposed in Section \ref{sec3} (Theorem \ref{thm1}) has an explicit convergence rate. Moreover, it directly assists with the hypothesis-testing problem formulated in this paper and provides a new perspective for future research on laws of large numbers. Since the channel features feedback, the input signals inherently constitute a predictable sequence presented in this paper. Additionally, non-i.i.d. input signals have broader potential applications, as it is challenging for real-world signals to satisfy the i.i.d. assumption commonly used in probability theory.

The rest of this paper is organized as follows. In Section \ref{sec2}, we review the basic notions and results in nonlinear expectation theory, and recall the large deviation theory required for the application. In Section \ref{sec3}, we establish a new law of large numbers with the rate of convergence, and provide a detailed proof. In Section \ref{sec4}, an application of our law of large numbers is given to the feedback channel-based detection problem.

\section{Preliminaries}\label{sec2}

\subsection{Nonlinear Expectation Theory}\label{sec2.1}

In this subsection, we present some concepts and existing results on nonlinear expectation theory. One can refer to \cite{Peng19} for more details about nonlinear expectation theory.

Let $(\Omega,\mathcal{F})$ be a given measurable space and $\mathcal{H}$ be a linear space of real-valued functions defined on $(\Omega,\mathcal{F})$ such that if $X_1,\cdots,X_n\in\mathcal{H}$, then $\varphi(X_1,\cdots,X_n)\in\mathcal{H}$ for all $\varphi\in B_b(\mathbb{R}^n)$ or $C_{l,Lip}(\mathbb{R}^n)$, where $B_b(\mathbb{R}^n)$ denotes the space of bounded Borel measurable functions on $\mathbb{R}^n$, and $C_{l,Lip}(\mathbb{R}^n)$ denotes the collection of function $f$ on $\mathbb{R}^n$ satisfying:
\begin{equation}
|f(\mathbf{x})-f(\mathbf{y})|\leq C(1+|\mathbf{x}|^m+|\mathbf{y}|^m)|\mathbf{x}-\mathbf{y}| \ \text{for}\ \mathbf{x},\mathbf{y}\in\mathbb{R}^n,\notag
\end{equation}
where $C>0$ and $m\in\mathbb{N}$ rely on $f$.
$C_{Lip}(\mathbb{R}^n)$ is denoted as the space of Lipschitz continuous functions on $\mathbb{R}^n$. 

\begin{definition}\label{def:sublinear expectation}
	A \textit{sublinear expectation} $\overline{\mathbb{E}}:\mathcal{H} \rightarrow \bar{\mathbb{R}}$ is a functional defined on the space $\mathcal{H}$ satisfying the following properties:
	\begin{itemize}
		\item[(i)] Monotonicity: $\overline{\mathbb{E}}[X]\leq\overline{\mathbb{E}}[Y]$ if $X\leq Y$;
		\item[(ii)] Constant preserving: $\overline{\mathbb{E}}[c]=c$ for $c\in\mathbb{R}$;
		\item[(iii)] Sub-additivity: $\overline{\mathbb{E}}[X+Y]\leq\overline{\mathbb{E}}[X]+\overline{\mathbb{E}}[Y]$ whenever $\overline{\mathbb{E}}[X]+\overline{\mathbb{E}}[Y]$ is not of the form $+\infty-\infty$ or $-\infty+\infty$;
		\item[(iv)] Positive homogeneity: $\overline{\mathbb{E}}[\lambda X]=\lambda\overline{\mathbb{E}}[X]$ for $\lambda\geq0$.
	\end{itemize}
	Here, $\bar{\mathbb{R}}\triangleq[-\infty,\infty]$ and $0\cdot\infty\triangleq 0$.
	
    The triplet $(\Omega,\mathcal{H},\overline{\mathbb{E}})$ is called a sublinear expectation space. If $(i)$ and $(ii)$ are satisfied, then $\overline{\mathbb{E}}$ is called a nonlinear expectation and $(\Omega,\mathcal{H},\overline{\mathbb{E}})$ is called a nonlinear expectation space.
\end{definition}

Let $\mathcal{M}$ be the collection of all probability
measures on $(\Omega,\mathcal{F})$. For a subset $\mathcal{P}\subseteq\mathcal{M}$, the upper expectation with respect to $\mathcal{P}$ is defined by
\begin{equation}\label{eq1}
\mathbb{E}[X]:=\sup_{P\in\mathcal{P}}E_P[X],\ X\in\mathcal{H}.
\end{equation}
The functional $\mathbb{E}$ defined by \eqref{eq1} is a sublinear expectation satisfying Definition~\ref{def:sublinear expectation}.

We say $\mathbf{X}=(X_1,\cdots,X_n)$ is a $n$-dimensional random vector in a sublinear expectation space $(\Omega,\mathcal{H},\mathbb{E})$ ($\mathbf{X}\in\mathcal{H}^n$ for short) if each coordinate of $\mathbf{X}$ is in $(\Omega,\mathcal{H},\mathbb{E})$. The following two definitions are the notions of identical distribution and independence for random variables under the framework of sublinear expectations in this paper.

\begin{definition}
	Let $(\Omega_1,\mathcal{H}_1,\mathbb{E}_1)$ and $(\Omega_2,\mathcal{H}_2,\mathbb{E}_2)$ be two sublinear expectation spaces. A n-dimensional random vector $\mathbf{X}_1$ in $(\Omega_1,\mathcal{H}_1,\mathbb{E}_1)$ is said to be identically distributed with another n-dimensional random vector $\mathbf{X}_2$ in $(\Omega_2,\mathcal{H}_2,\mathbb{E}_2)$, denoted by $\mathbf{X}_1\overset{d}{=}\mathbf{X}_2$, if for all $\varphi\in C_{l,Lip}(\mathbb{R}^n)$,
	\begin{equation}
	\mathbb{E}_1[\varphi(\mathbf{X}_1)]=\mathbb{E}_2[\varphi(\mathbf{X}_2)].\notag
	\end{equation}
	
	A sequence of random variables $\{X_n\}_{n=1}^\infty$ in $(\Omega,\mathcal{H},\mathbb{E})$ is said to be identically distributed if $X_i\overset{d}{=}X_1$ for each $i\geq 1$.
\end{definition}

\begin{definition}
	Let $(\Omega,\mathcal{H},\mathbb{E})$ be a sublinear expectation space. A random vector $\mathbf{Y}=(Y_1,\cdots,Y_n)\in \mathcal{H}^n$ is said to be independent of another random vector $\mathbf{X}=(X_1,\cdots,X_m)\in\mathcal{H}^m$ under $\mathbb{E}$ if
	\begin{equation}
	\mathbb{E}[\varphi(\mathbf{X},\mathbf{Y})]=\mathbb{E}[\mathbb{E}[\varphi(\mathbf{x},\mathbf{Y})]|_{\mathbf{x}=\mathbf{X}}],\notag
	\end{equation}
	for each function $\varphi\in C_{l,Lip}(\mathbb{R}^{m+n})$ such that $\hat{\varphi}(\mathbf{x}):=\mathbb{E}[|\varphi(\mathbf{x}, \mathbf{Y})|]<\infty$ for all $\mathbf{x}\in\mathbb{R}^m$, and $\mathbb{E}[\varphi(\mathbf{x},\mathbf{Y})]|_{\mathbf{x}=\mathbf{X}}\in\mathcal{H}$, $\hat{\varphi}(\mathbf{X})\in\mathcal{H}$, $\mathbb{E}[\hat{\varphi}(\mathbf{X})]<\infty$.
	
	A sequence of random variables $\{X_n\}_{n=1}^\infty$ in $(\Omega,\mathcal{H},\mathbb{E})$ is said to be independent if $X_{i+1}$ is independent of $(X_1,\cdots,X_i)$ for each $i\geq 1$.
\end{definition}
It is worth noting that under the framework of sublinear expectations, a random variable $Y$ is independent of another random variable $X$ can not in general imply that $X$ is independent of $Y$, which is different from traditional independence in classical probability theory. One can refer to the Example 1.3.15 in \cite{Peng19} for more details. The maximal distribution generally serves as the limit distribution of the LLN under sublinear expectations and its original definition is Definition 2.2.1 in \cite{Peng19}. The following is the definition of maximal distribution used in this paper.
\begin{definition}
	A random variable $X$ on a sublinear expectation space $(\Omega,\mathcal{H},\mathbb{E})$ is called maximally distributed if there exist $a,b\in\mathbb{R}$ such that 
	\begin{equation}
		\mathbb{E}[\varphi(X)]=\max_{r\in[a,b]}\varphi(r),\ \varphi\in C_{Lip}(\mathbb{R}).\notag
	\end{equation}
\end{definition}

The following is a LLN with a convergence rate under a sublinear expectation, as established in \cite{Hu21}.
\begin{theorem}\label{th1}
	Let $\{X_i\}_{i=1}^\infty$ be the independent random variables on a sublinear expectation space $(\Omega,\mathcal{H},\mathbb{E})$, with $\mathbb{E}[X_i]=\overline{\mu}$ and $-\mathbb{E}[-X_i]=\underline{\mu}$ for $i\geq 1$. Let $S_n=X_1+\cdots+X_n$ for each $n\geq 1$. Assume there exists $\alpha\in (0,1]$ such that $C_{\alpha}:=\underset{i\geq 1}{\sup}\ \mathbb{E}[|X_i|^{1+\alpha}]<\infty$. Then, for each $\varphi\in C_{Lip}(\mathbb{R})$ with Lipschitz constant $L_{\varphi}$, there is
	\begin{equation}
	\left|\mathbb{E}\left[\varphi\left(\frac{S_n}{n}\right)\right]-\max_{r\in[\underline{\mu},\overline{\mu}]}\varphi(r)\right|\leq L_{\varphi}\left(\frac{4C_{\alpha}}{n^\alpha}\right)^{\frac{1}{1+\alpha}} \ \text{for each}\ n\geq 1.
	\end{equation}
\end{theorem}

\subsection{Large Deviations}\label{sec2.2}
For a random variable $X$ in a probability space $(\Omega,\mathcal{F},P)$ with $\overline{x}:=E_P[X]<\infty$, the logarithmic moment generating function of $X$ is defined as 
\begin{equation}
\Lambda(\lambda):=\log E_P[e^{\lambda X}],\ \lambda\in\mathbb{R}.\notag
\end{equation}
The Fenchel-Legendre transform of $\Lambda(\lambda)$ is
\begin{equation}
\Lambda^*(x):=\sup_{\lambda\in\mathbb{R}}\{\lambda x-\Lambda(\lambda)\},\ x\in\mathbb{R}.\notag
\end{equation}
The following are some properties of the Fenchel-Legendre transform $\Lambda^*(x)$ in Lemma 2.2.5 in \cite{Dembo09}.
\begin{lemma}\label{le3}
	\begin{itemize}
		\item[(i)]$\Lambda^*(\overline{x})=\underset{x\in\mathbb{R}}{\inf}\Lambda^*(x)=0$.
		\item[(ii)]If $\Lambda(\lambda)<\infty$ for some $\lambda>0$, then for all $x>\overline{x}$,
		\begin{equation}
		\Lambda^*(x)=\sup_{\lambda\geq0}\{\lambda x-\Lambda(\lambda)\},\notag
		\end{equation}
		and $\Lambda^*(x)$ is a nondecreasing function on $(\overline{x},\infty)$.
		\item[(iii)]If $\Lambda(\lambda)<\infty$ for some $\lambda<0$, then for all $x<\overline{x}$,
		\begin{equation}
		\Lambda^*(x)=\sup_{\lambda\leq0}\{\lambda x-\Lambda(\lambda)\},\notag
		\end{equation}
		and $\Lambda^*(x)$ is a nonincreasing function on $(-\infty,\overline{x})$.
	\end{itemize}
\end{lemma}
\begin{lemma}\label{le4}
	Let $\{X_i\}_{i=1}^\infty$ be an i.i.d. sequence with $\overline{x}:=E_P[X_1]<\infty$ in $(\Omega,\mathcal{F},P)$. Let $\Lambda(\lambda)$ be the logarithmic moment generating function of $X_1$ and $\Lambda^*(x)$ be the Fenchel-Legendre transform of $\Lambda(\lambda)$. Set $W_n:=\frac{X_1+\cdots+X_n}{n}$. If $\Lambda(\lambda)<\infty$ for some $\lambda>0$, then for all $x>\overline{x}$,
	\begin{equation}
	P(W_n\geq x)\leq e^{-n\Lambda^*(x)}.\notag
	\end{equation}  
\end{lemma}
\begin{proof}
	For all $\lambda\geq 0$ and $x> \overline{x}$, 
	\begin{equation}
	P(W_n\geq x) = P(e^{\lambda nW_n}\geq e^{\lambda nx})\leq e^{-\lambda nx}E_P[e^{\lambda nW_n}]=e^{-n(\lambda x-\Lambda(\lambda))}.\notag
	\end{equation}
	Hence,
	\begin{equation}
	P(W_n\geq x)\leq \inf_{\lambda\geq0}e^{-n(\lambda x-\Lambda(\lambda))}=e^{-n\cdot \underset{\lambda\geq 0}{\sup} \{\lambda x-\Lambda(\lambda)\}}.\notag
	\end{equation}
	By (ii) in Lemma \ref{le3}, we know that
	\begin{equation}
	\sup_{\lambda\geq0}\{\lambda x-\Lambda(\lambda)\}=\Lambda^*(x).\notag
	\end{equation}
	Therefore, the proof is completed.
\end{proof}

\section{A Law of Large Numbers Based on Nonlinear Expectation Theory}\label{sec3}

We establish a new LLN for special partial sums, with its proof inspired by Theorem 4.1 in \cite{Fang19}. Our theorem provides a LLN with the convergence rate in a probability space, whereas Theorem 4.1 in \cite{Fang19} pertains to a CLT with a convergence rate under the probabilistic setting as well. Theorem \ref{thm1} also serves as a foundation for analyzing and addressing the convergence rate of the detection error probability in hypothesis testing, which will be discussed in Section \ref{sec4}.

\begin{theorem}\label{thm1}
	Let $\{Z_i\}_{i=1}^\infty$ and $\{K_i\}_{i=1}^\infty$ be two i.i.d. sequences of random variables in a filtered probability space $(\Omega,\mathcal{F},P,\{\tilde{\mathcal{F}}_i\}_{i=0}^\infty)$, where $\tilde{\mathcal{F}}_0:=\sigma(\emptyset,\Omega)$, $\tilde{\mathcal{F}}_i:=\sigma(Z_1,\cdots,Z_i,K_1,\cdots,K_i)$, $i=1,2,\cdots$. We assume that $E[|Z_1|^2]=\sigma^2>0$, $E[K_1]=0$, $E[|K_1|^2]=k^2>0$, $E[|Z_1|^{2+2\alpha}]\vee E[|K_1|^{2+2\alpha}]<\infty$ for a given $\alpha\in(0,1]$,  and $\{Z_i\}_{i=1}^\infty$ is independent of $\{K_i\}_{i=1}^\infty$, where $E$ denotes the expectation with respect to the probability measure $P$. Let $\tilde{\Sigma}[\underline{\zeta} ,\overline{\zeta}]$ be the collection of all predictable sequences with respect to $\{\tilde{\mathcal{F}}_i\}_{i=0}^\infty$ that take value in $[\underline{\zeta},\overline{\zeta}]$, where $0<\underline{\zeta}<\overline{\zeta}$ are two given constants. For any sequence $\zeta:=\{\zeta_i\}_{i=1}^\infty\in\tilde{\Sigma}[\underline{\zeta} ,\overline{\zeta}]$ and $n\in\mathbb{N}$, define $Z_i^\zeta:=(\zeta_iZ_i+K_i)^2$  and $W_{i,n}^\zeta:=\frac{Z_1^\zeta+\cdots+Z_i^\zeta}{n}$ for $i=1,\cdots,n$. Then for any $\varphi\in C_{Lip}(\mathbb{R})$ with Lipschitz constant $L_\varphi$, we have
		\begin{equation}\label{eq4}
		\Bigg|\sup_{\zeta\in\tilde{\Sigma}[\underline{\zeta},\overline{\zeta}]}E[\varphi(W_{n,n}^\zeta)]-\max_{r\in[\sigma^2\underline{\zeta}^2+k^2,\sigma^2\overline{\zeta}^2+k^2]}\varphi(r)\Bigg|\leq L_\varphi\left(\frac{4\tilde{C}_\alpha}{n^\alpha}\right)^{\frac{1}{1+\alpha}},
		\end{equation}
		where $\tilde{C}_\alpha:=\underset{\zeta\in\tilde{\Sigma}[\underline{\zeta},\overline{\zeta}]}{\sup} E[|Z_1^\zeta|^{1+\alpha}] \leq 8\overline{\zeta}^{2+2\alpha} E[|Z_1|^{2+2\alpha}]+8E[|K_1|^{2+2\alpha}]<\infty$.
\end{theorem}

\begin{proof}
	For fixed $n$, define $\xi_i:\mathbb{R}^n\rightarrow\mathbb{R}$ by $\xi_i(x_1,\cdots,x_n)=x_i,\ i=1,\cdots,n$. Denote $\tilde{\mathcal{H}}$ as the collection of real-valued continuous functions $\tilde{h}$ on $\mathbb{R}^n$ with $|\tilde{h}(x)|\leq C(1+|x|^{1+\alpha})$. Define a sublinear expectation as follows:
	\begin{equation}
	\tilde{\mathbb{E}}[\tilde{h}]:=\sup_{\zeta\in\tilde{\Sigma}[\underline{\zeta} ,\overline{\zeta}]}E[\tilde{h}(Z_1^\zeta,\cdots,Z_n^\zeta)],\ \forall \tilde{h}\in\tilde{\mathcal{H}}. \notag
	\end{equation}
	For $i=1,\cdots,n$, define the distribution of $\xi_i$ under $\tilde{\mathbb{E}}$ as follows:
	\begin{equation}
	\tilde{\mathcal{N}_i}[\varphi]:=\tilde{\mathbb{E}}[\varphi(\xi_i)],\ \forall\varphi(x)\in C(\mathbb{R}) \ \text{with} \ |\varphi(x)|\leq C(1+|x|^{1+\alpha}) .\notag
	\end{equation}
	Because $\{Z_i\}_{i=1}^\infty$ and $\{K_i\}_{i=1}^\infty$ are two i.i.d. sequences and $\zeta$ is a predictable sequence with respect to the filtration $\{\tilde{\mathcal{F}}_i\}_{i=0}^\infty$, we have for all $i=1,\cdots,n$ and $\varphi\in C(\mathbb{R})$ with $|\varphi(x)|\leq C(1+|x|^{1+\alpha})$,
	\begin{equation}
	\tilde{\mathcal{N}_i}[\varphi]=\sup_{x\in[\underline{\zeta}^2,\overline{\zeta}^2]}E[\varphi((xZ_1+K_1)^2)].\notag
	\end{equation}
	Therefore, $\tilde{\mathbb{E}}[\xi_i]=\underset{\zeta\in\tilde{\Sigma}[\underline{\zeta},\overline{\zeta}]}{\sup}E[Z_i^\zeta]=\sigma^2\overline{\zeta}^2+k^2$, $-\tilde{\mathbb{E}}[-\xi_i]=\underset{\zeta\in\tilde{\Sigma}[\underline{\zeta},\overline{\zeta}]}{\inf}E[Z_i^\zeta]=\sigma^2\underline{\zeta}^2+k^2$, and $\tilde{\mathbb{E}}[|\xi_i|^{1+\alpha}]=\underset{\zeta\in\tilde{\Sigma}[\underline{\zeta},\overline{\zeta}]}{\sup}E[|Z_1^\zeta|^{1+\alpha}]\leq8\overline{\zeta}^{2+2\alpha} E[|Z_1|^{2+2\alpha}]+8E[|K_1|^{2+2\alpha}]<\infty$.
	\par Set $W_i^n:=\frac{\xi_1+\cdots\xi_i}{n},\ i=1,\cdots,n$. Due to the predictability of $\zeta$, we know that for each $\zeta_i,\ i=2,\cdots,n$, there exists a Borel measurable function $g_i$ on $\mathbb{R}^{2i-2}$ such that $\zeta_i=g_i(Z_1,\cdots,Z_{i-1},K_1,\cdots,K_{i-1})$. It is not vague that we still use $\zeta_i$ to represent $g_i$, which means \begin{equation}
	\zeta_i(Z_1,\cdots,Z_{i-1},K_1,\cdots,K_{i-1})=g_i(Z_1,\cdots,Z_{i-1},K_1,\cdots,K_{i-1}),\ i=2,\cdots,n.\notag
	\end{equation}
	\par For each fixed $i$ and $\zeta\in\tilde{\Sigma}[\underline{\zeta},\overline{\zeta}]$, define $\bm{x}_{2i}=(x_1,\cdots,x_{2i})$ and
	\begin{equation}
	f^{\zeta,n}(\bm{x}_{2i}):=\frac{(\zeta_1x_1+x_{i+1})^2}{n}+\sum_{j=2}^{i}\frac{(\zeta_j(x_1,\cdots,x_{j-1},x_{i+1},\cdots,x_{i+j-1})x_j+x_{i+j})^2}{n}\text{on}\ \mathbb{R}^{2i}.\notag
	\end{equation}
	\par Next, we prove that for any $\varphi\in C_{Lip}(\mathbb{R})$ and $i=1,\cdots,n-1$,
	\begin{equation}
	\tilde{\mathbb{E}}[\varphi(W_{i+1}^n)]=\tilde{\mathbb{E}}\left[ \tilde{\mathbb{E}}\left[ \varphi\left(m+\frac{\xi_{i+1}}{n}\right)\right]\bigg|_{m=W_i^n} \right].\label{eq18}    
	\end{equation}
	On the one hand, for any $\zeta\in\tilde{\Sigma}[\underline{\zeta},\overline{\zeta}]$ and $i=1,\cdots,n-1$,
	\begin{align*}
	E[\varphi(W_{i+1,n}^\zeta)]
	&=E\left[\varphi\left(\frac{(\zeta_1Z_1+K_1)^2}{n}+\cdots+\frac{(\zeta_iZ_i+K_i)^2}{n}+\frac{(\zeta_{i+1}Z_{i+1}+K_{i+1})^2}{n} \right)  \right] \\
	&=E\left[ E\left[\varphi\left(f^{\zeta,n}(\bm{x}_{2i})+\frac{(\zeta_{i+1}(\bm{x}_{2i})Z_{i+1}+K_{i+1})^2}{n} \right)  \right]\bigg|_{\bm{x}_{2i}=(Z_1,\cdots,Z_i,K_1,\cdots,K_i)} \right]\\
	&\leq E\left[ \sup_{\lambda\in[\underline{\zeta},\overline{\zeta}]}E\left[\varphi\left(f^{\zeta,n}(\bm{x}_{2i})+\frac{(\lambda Z_{i+1}+K_{i+1})^2}{n} \right)  \right]\bigg|_{\bm{x}_{2i}=(Z_1,\cdots,Z_i,K_1,\cdots,K_i)} \right]\\
	&=E\left[\tilde{\mathbb{E}}\left[\varphi\left(f^{\zeta,n}(\bm{x}_{2i})+\frac{\xi_{i+1}}{n} \right) \right]\bigg|_{\bm{x}_{2i}=(Z_1,\cdots,Z_i,K_1,\cdots,K_i)}  \right]\\
	&=E\left[\tilde{\mathbb{E}}\left[\varphi\left(m+\frac{\xi_{i+1}}{n} \right) \right]\bigg|_{m=W_{i,n}^\zeta}  \right]\\
	&\leq\tilde{\mathbb{E}}\left[\tilde{\mathbb{E}}\left[\varphi\left(m+\frac{\xi_{i+1}}{n} \right) \right]\bigg|_{m=W_i^n}\right].
	\end{align*}
	Hence, we have
	\begin{equation}
	\tilde{\mathbb{E}}[\varphi(W_{i+1}^n)]\leq\tilde{\mathbb{E}}\left[\tilde{\mathbb{E}}\left[\varphi\left(m+\frac{\xi_{i+1}}{n} \right) \right]\bigg|_{m=W_i^n}\right].\label{eq2}
	\end{equation}
	On the other hand, we utilize the technique of measurable approximations in the proof of Theorem 4.1 in \cite{Fang19}. 
	For each $m\in\mathbb{R}$, we can find $\lambda^{\varphi,n}(m)\in[\underline{\zeta},\overline{\zeta}]$ such that
	\begin{equation}
	E\left[ \varphi\left( m+\frac{(\lambda^{\varphi,n}(m)Z_{1}+K_1)^2}{n}\right) \right] =\sup_{\lambda\in [\underline{\zeta},\overline{\zeta}]}E\left[ \varphi\left( m+\frac{(\lambda Z_{1}+K_1)^2}{n}\right) \right] =\tilde{\mathbb{E}}\left[ \varphi\left( m+\frac{\xi_{1}}{n}\right) \right].\notag
	\end{equation}
	Then, define $\Lambda(s,t,Z_1,K_1):=\varphi\left( s+\frac{(\lambda^{\varphi,n}(t)Z_{1}+K_1)^2}{n}\right) $, $\forall s,t\in\mathbb{R}$.
	For any $s,t\in \mathbb{R}$, we have
	\begin{align*} 
	&E[\Lambda(s,s, Z_1,K_1)]\\
	=&E[\Lambda(t,s, Z_1,K_1)]
	+\big(E[\Lambda(s,s, Z_1,K_1)]-E[\Lambda(t,s, Z_1,K_1)]\big)\\
	\leq&E[\Lambda(t,t, Z_1,K_1)]+ L_\varphi|t-s|\\
	=&E[\Lambda(s,t, Z_1,K_1)]
	+\big(E[\Lambda(t,t, Z_1,K_1)]-E[\Lambda(s,t, Z_1,K_1)]\big)+L_\varphi|t-s|\\    
	\leq&E[\Lambda(s,t, Z_1,K_1)]+2L_\varphi|t-s|,
	\end{align*}
	where the first inequality is by the definition of $\lambda^{\varphi,n}$, and $L_\varphi$ is the Lipschitz constant of the function $\varphi$.
	For any $\epsilon>0$, set $\delta=\frac{\epsilon}{2L_\varphi}$ and
	\[\lambda^{\varphi,n}_\epsilon(m):=\sum_{k\in\mathbb{Z}}\lambda^{\varphi,n}(k\delta)1_{[k\delta,(k+1)\delta)}(m).\]
	Then, for any $m\in \mathbb{R}$,
\begin{align*}
	E\left[ \varphi\left( m+\frac{(\lambda^{\varphi,n}_\epsilon(m)Z_{1}+K_1)^2}{n}\right) \right]
	&=\sum_{k\in\mathbb{Z}}E\left[ \varphi\left( m+\frac{(\lambda^{\varphi,n}(k\delta)Z_{1}+K_1)^2}{n}\right) \right]1_{[k\delta,(k+1)\delta)}(m)\\
	&=\sum_{k\in\mathbb{Z}}E[\Lambda(m,k\delta, Z_1,K_1)]1_{[k\delta,(k+1)\delta)}(m)\\
	&\geq\sum_{k\in\mathbb{Z}}E[\Lambda(m,m, Z_1,K_1)]1_{[k\delta,(k+1)\delta)}(m)-\epsilon\\
	&=\tilde{\mathbb{E}}\left[ \varphi\left( m+\frac{\xi_{1}}{n}\right) \right] -\epsilon.
\end{align*}
	
	\par For any $\tilde{\zeta}\in\tilde{\Sigma}[\underline{\zeta},\overline{\zeta}]$ with $\tilde{\zeta}_{i+1}(\bm{x}_{2i})=\lambda^{\varphi,n}_\epsilon(f^{\tilde{\zeta},n}(\bm{x}_{2i}))$, we have
	\begin{align*}
	E[\varphi(W^{\tilde{\zeta}}_{i+1,n})]
	&=E\left[ E\left[ \varphi\left( f^{\tilde{\zeta},n}(\bm{x_{2i}})+\frac{(\lambda^{\varphi,n}_\epsilon(f^{\tilde{\zeta},n}(\bm{x}_{2i}))Z_{i+1}+K_{i+1})^2}{n}\right) \right] \bigg|_{\bm{x_{2i}}=(Z_1,\cdots,Z_i,K_1,\cdots,K_i)}\right] \\
	&=E\left[ E\left[ \varphi\left( m+\frac{(\lambda^{\varphi,n}_\epsilon(m)Z_{1}+K_1)^2}{n}\right) \right] \bigg|_{m=W^{\tilde{\zeta}}_{i,n}}\right] \\
	&\geq E\left[ \tilde{\mathbb{E}}\left[ \varphi\left( m+\frac{\xi_{i+1}}{n}\right) \right] \bigg|_{m=W^{\tilde{\zeta}}_{i,n}}\right] -\epsilon.
	\end{align*}
	Therefore,
	\begin{equation}
	\tilde{\mathbb{E}}[ \varphi(W_{i+1}^n))]\geq\tilde{\mathbb{E}}\left[ \tilde{\mathbb{E}}\left[ \varphi\left( m+\frac{\xi_{i+1}}{n}\right) \right] \bigg|_{m=W_i^n}\right].\label{eq3}
	\end{equation} 
	Combining \eqref{eq2} and \eqref{eq3}, we obtain \eqref{eq18}. Let $\hat{\xi}_1,\cdots,\hat{\xi}_n$ be i.i.d. random variables under a sublinear expectation $\hat{\mathbb{E}}$ and $\hat{\xi}_1\overset{d}{=}\xi_1$. By \eqref{eq18}, we have, for any $\varphi\in C_{Lip}(\mathbb{R})$,
	\begin{equation}
	\tilde{\mathbb{E}}[\varphi(W_n^n)]=\hat{\mathbb{E}}\left[ \varphi\left( \frac{\hat{\xi}_1+\cdots+\hat{\xi}_n}{n}\right) \right].\notag
	\end{equation}
	Note that $\hat{\mathbb{E}}[\hat{\xi}_i]=\tilde{\mathbb{E}}[\xi_1]=\sigma^2\overline{\zeta}^2+k^2$, $-\hat{\mathbb{E}}[-\hat{\xi}_i]=-\tilde{\mathbb{E}}[-\xi_1]=\sigma^2\underline{\zeta}^2+k^2$, and $\hat{\mathbb{E}}[|\hat{\xi}_i|^{1+\alpha}]=\tilde{\mathbb{E}}[|\xi_1|^{1+\alpha}]<\infty$, for $i=1,\cdots,n$. Then, by Theorem \ref{th1}, we obtain \eqref{eq4}. This concludes the proof.
\end{proof}

Intuitively, the uncertainty in the bounded predictable sequences in Theorem \ref{thm1} leads to uncertainty in the limiting distribution. If the bounded predictable sequences reduce to a single point, the theorem degenerates to its classical counterpart. In Theorem \ref{thm1}, since $\{Z_i^2\}_{i=1}^\infty$ is a non-negative sequence with a positive mean, under the influence of the uncertain bounded predictable sequence $\{\zeta_i^2\}_{i=1}^\infty$, the distributions of the means of the partial sums $W_{n,n}^\zeta$ converge to a maximal distribution with an uncertain mean. Furthermore, because the sequence $\{K_i\}_{i=1}^\infty$ is centered (i.e., each $K_i$ has zero mean), the uncertainty from $\{\zeta_i\}_{i=1}^\infty$ is offset in its influence on the sequence $\{K_iZ_i\}_{i=1}^\infty$.

The Marcinkiewicz strong law of large numbers in classical probability theory tells us that when random variables satisfy higher integrability conditions, the strong law of large numbers can achieve a faster convergence rate. The law of large numbers presented in this paper also exhibits this pattern: when the random variables have integrability of order $2+2\alpha$, $\alpha\in(0,1]$, the law of large numbers has the convergence rate $O(1/n^{\frac{\alpha}{1+\alpha}})$.

\section{An Application to Detection Problem in Communication Systems}\label{sec4}
We give an application of Theorem \ref{thm1} to the feedback channel-based detection problem, which arises naturally in modern communication systems where the transmitted symbols are adaptively determined through a feedback channel. In such systems, the receiver must decide whether a transmission has been activated based on a sequence of received signals that are statistically dependent due to the feedback mechanism. In this setting, the input and output sequences are generally non-i.i.d. The classical i.i.d. assumption, while analytically convenient, may not fully capture the feedback-induced dependencies that often occur in practical communication scenarios. For the classical detection problem, readers may refer to \cite{Urkowitz1967energy,MF2012,Benitez2013signal}.

Our law of large numbers with the convergence rate provides a rigorous tool to characterize how fast the detection error probabilities decay with the number of observations $n$ in this non-i.i.d. environment. This application demonstrates the practical relevance of our theoretical results in a context where dependence and feedback are intrinsic. Such feedback channel-based detection arises in adaptive sensing, cognitive radio, and uplink activation in modern wireless systems, where understanding the rate of convergence directly relates to the design of reliable decision rules under realistic system dynamics.

\subsection{Models}\label{sec4.1}
Let $(\Omega,\mathcal{H},\mathbb{E})$ be a given sublinear expectation space, and let $\mathcal{P}$ be the family of probability measures representing $\mathbb{E}$ as expressed in \eqref{eq1}. We first introduce a set of assumptions on the random variables associated with the channel.

Let $n$ denotes the number of samples collected during the channel observations. Under each $P\in \mathcal{P}$, let $\epsilon:=\{\epsilon_i\}_{i=1}^{n}$ be the channel fading and $\{\epsilon_i\}_{i=1}^{n}$ is a sequence of i.i.d. random variables with $\mathbb{E}[\epsilon_1^2]=\overline{\sigma}^2, -\mathbb{E}[-\epsilon_1^2]=\underline{\sigma}^2>0$ and $\mathbb{E}[|\epsilon_1|^{2+2\alpha}]<\infty$ for some $\alpha\in(0,1]$. Let $\eta:=\{\eta_i\}_{i=1}^n$ be the existing channel noise sequence in the channel and $\{\eta_i\}_{i=1}^n$ is an i.i.d. sequence under each $P\in\mathcal{P}$, and satisfies $\mathbb{E}[\eta_1]=\mathbb{E}[-\eta_1]=0$, $\mathbb{E}[\eta_1^2]=\overline{\nu}^2$, $-\mathbb{E}[-\eta_1^2]=\underline{\nu}^2>0$ and $\mathbb{E}[|\eta_1|^{2+2\alpha}]<\infty$. Let $X:=\{X_i\}_{i=1}^n$ be the input random variables existing potentially in the channel and $X_i$ takes value in $[\underline{\zeta} ,\overline{\zeta}]$ for all $i=1,\cdots,n$, where $\underline{\zeta},\overline{\zeta}$ are two given positive constants ($0<\underline{\zeta}<\overline{\zeta}<\infty$).

Let $\mathcal{F}_0:=\sigma(\emptyset,\Omega)$, $\mathcal{F}_i:=\sigma(\epsilon_1,\cdots,\epsilon_i,\eta_1,\cdots,\eta_i) \subset \mathcal{F} ,i=1,\cdots,n$. In a feedback channel-based transmission scheme, the transmitted signal $X_i$ can depend on the past realizations of the channel. Therefore, we assume $X_i\in\mathcal{F}_{i-1},i=1,\cdots,n$, which means $\{X_i\}_{i=1}^n$ is a predictable sequence with respect to $\{\mathcal{F}_i\}_{i=0}^n$. This structure also induces dependence among the received signals, resulting in a non-i.i.d. observation process. Moreover, we assume that $\{\eta_i\}_{i=1}^n$ is independent of $\{\epsilon_i\}_{i=1}^n$ under each $P\in\mathcal{P}$, which means the channel fading and the noise are independent of each other. Let $\Sigma[\underline{\zeta} ,\overline{\zeta}]$ be the collection of all predictable sequences with respect to $\{\mathcal{F}_i\}_{i=0}^n$ that take value in $[\underline{\zeta},\overline{\zeta}]$. Let $Y:=\{Y_i\}_{i=1}^n$ be the output random variables. All these random variables are assumed to be in $\mathcal{H}$.

The feedback channel-based detection problem we consider in this paper can be formulated as a binary hypothesis-testing problem:
\begin{equation}
\begin{split}
&\mathcal{H}_0:Y_i=\epsilon_i\cdot X_i+\eta_i,\ \quad  i=1,\cdots,n,\\
&\mathcal{H}_1:Y_i=\eta_i,\ \ \ \ \ \ \ \quad \quad \ \ i=1,\cdots,n,
\end{split}
\end{equation}
The receiver observes $\{Y_i\}_{i=1}^n$ and decides between $\mathcal{H}_0$ and $\mathcal{H}_1$ based on the accumulated received energy $T_n = \sum_{i=1}^{n}Y_i^2$. A decision rule of the form
\begin{equation}
\Psi:\ \text{If} \ T_n \leq \lambda,\ \text{then reject the null hypothesis.}\notag
\end{equation}
is employed, where $\lambda$ is a detection threshold.

In this hypothesis testing framework, a \textit{missed detection} occurs when the transmitter is active and sends nonzero symbols through the channel, but the detector incorrectly selects the hypothesis $ \mathcal{H}_1$. On the other hand, a \textit{false alarm} occurs when the transmitter is inactive, but the detector incorrectly selects the hypothesis $ \mathcal{H}_0$. Therefore, it is necessary to analyze the probabilities of these two events. We use the supremum to dominate the error probabilities uniformly, that is, to dominate the case with the maximum of error probabilities under the uncertain input random variables in $\Sigma[\underline{\zeta} ,\overline{\zeta}]$ and the probability measure in $\mathcal{P}$ (the worst case). We consider the \textit{upper probability of missed detections}, that is $\underset{P\in\mathcal{P}}{\sup} \underset{X\in\Sigma[\underline{\zeta} ,\overline{\zeta}]}{\sup} P(\text{choosing}\ H_1)_{H_0}$, and the \textit{upper probability of false alarms}, i.e., $\underset{P\in\mathcal{P}}{\sup} \underset{X\in\Sigma[\underline{\zeta} ,\overline{\zeta}]}{\sup} P(\text{choosing}\ H_0)_{H_1}$. This motivates the application of our law of large numbers with the convergence rate to analyze how the error probabilities decay with $n$ under dependent observations.

\subsection{The Convergence Rates of Detection Error Probabilities}\label{sec4.2}

\begin{theorem}\label{th2}
	Under the setting in Subsection \ref{sec4.1}, given $p>0$, then a statistical test with the upper probability of missed detections not exceeding $p$ is as follows:
		\begin{equation}
		\Psi: \text{If} \  \sum_{i=1}^{n}Y_i^2 \leq n \gamma_n,\ \text{then we reject the null hypothesis.}\notag
		\end{equation}
		where 
		\begin{align}
		&\gamma_n:=\underline{\sigma}^2\underline{\zeta}^2+\underline{\nu}^2-\frac{C_{\alpha}}{p\cdot n^{\frac{\alpha}{1+\alpha}}},\label{eq19}\\
		&C_{\alpha}:=(32\overline{\zeta}^{2+2\alpha} \mathbb{E}[|\epsilon_1|^{2+2\alpha}]+32 \mathbb{E}[|\eta_1|^{2+2\alpha}])^{\frac{1}{1+\alpha}}.\label{eq20}
		\end{align}
\end{theorem}

\begin{proof}
	First, We consider the statistical test $\Psi$ as the following form:
	\begin{equation}
	\Psi: \ \text{If}\ \ \ \sum_{i=1}^{n}Y_i^2\leq n\gamma,\notag
	\end{equation}
	then reject $H_0$, or accept (rejecting $H_1$) otherwise. Here, $\gamma$ is the predefined threshold. Next, we shall determine $\gamma$ in terms of the level of the upper probability of missed detections.
	\par Note that the upper probability of missed detections is
	\begin{equation}
	\sup_{P\in\mathcal{P}}\sup_{X\in\Sigma[\underline{\zeta},\overline{\zeta }]}P\left( \frac{\sum_{i=1}^{n}Y_i^2}{n}\leq \gamma\right)_{H_0} = \sup_{P\in\mathcal{P}}\sup_{X\in\Sigma[\underline{\zeta},\overline{\zeta}]}P\left( \frac{\sum_{i=1}^{n}(\epsilon_iX_i+\eta_i)^2}{n}\leq \gamma\right).\notag
	\end{equation}
	Let $\gamma<\underline{\sigma}^2\underline{\zeta}^2+\underline{\nu}^2$, and define 
	\begin{equation}
	\varphi_1(x)=\left\{
	\begin{aligned}
	&1,&&x\leq  \gamma,\\
	&-\frac{1}{\underline{\sigma}^2\underline{\zeta}^2+\underline{\nu}^2-\gamma}x+1+\frac{\gamma}{\underline{\sigma}^2\underline{\zeta}^2+\underline{\nu}^2- \gamma},&& \gamma<x\leq \underline{\sigma}^2\underline{\zeta}^2+\underline{\nu}^2,\\
	&0,&&x>\underline{\sigma}^2\underline{\zeta}^2+\underline{\nu}^2.\notag
	\end{aligned}
	\right.
	\end{equation}
	It's clear that $\varphi_1(x)\geq I_{(-\infty,\gamma]}(x)$, $x\in\mathbb{R}$ and $\varphi_1=0$ on $[\underline{\sigma}^2\underline{\zeta}^2+\underline{\nu}^2,\overline{\sigma}^2\overline{\zeta}^2+\overline{\nu}^2]$. For each fixed $P\in\mathcal{P}$, noting that $E_P[\epsilon_1^2]\underline{\zeta}^2+E_P[\eta_1^2] \geq \underline{\sigma}^2\underline{\zeta}^2+\underline{\nu}^2$ and $E_P[\epsilon_1^2]\overline{\zeta}^2+E_P[\eta_1^2] \leq \overline{\sigma}^2\overline{\zeta}^2+\overline{\nu}^2$, letting $\varphi=\varphi_1$ and $\zeta_i=X_i$, $Z_i=\epsilon_i$, $K_i=\eta_i$ for $i=1,\cdots,n$ in Theorem \ref{thm1}, we have  
	\begin{align}
	&\sup_{X\in\Sigma[\underline{\zeta},\overline{\zeta}]}E_P\left[I_{(-\infty,\gamma]}\left(\frac{\sum_{i=1}^{n}(\epsilon_iX_i+\eta_i)^2}{n} \right)\right]\notag\\
	\leq&\sup_{X\in\Sigma[\underline{\zeta},\overline{\zeta}]}E_P\left[\varphi_1\left(\frac{\sum_{i=1}^{n}(\epsilon_iX_i+\eta_i)^2}{n} \right)\right]\notag\\
	\leq&\frac{1}{\underline{\sigma}^2\underline{\zeta}^2+\underline{\nu}^2-\gamma}\left(\frac{32\overline{\zeta}^{2+2\alpha} E_P[|\epsilon_1|^{2+2\alpha}]+32 E_P[|\eta_1|^{2+2\alpha}]}{n^\alpha}\right)^{\frac{1}{1+\alpha}}\notag\\
	\leq&\frac{1}{\underline{\sigma}^2\underline{\zeta}^2+\underline{\nu}^2-\gamma}\left(\frac{32\overline{\zeta}^{2+2\alpha} \mathbb{E}[|\epsilon_1|^{2+2\alpha}]+32 \mathbb{E}[|\eta_1|^{2+2\alpha}]}{n^\alpha}\right)^{\frac{1}{1+\alpha}}.\notag
	\end{align}
	Therefore, 
	\begin{equation}
	\sup_{P\in\mathcal{P}}\sup_{X\in\Sigma[\underline{\zeta},\overline{\zeta }]}P\left( \frac{\sum_{i=1}^{n}Y_i^2}{n}\leq \gamma\right)_{H_0} \leq \frac{1}{\underline{\sigma}^2\underline{\zeta}^2+\underline{\nu}^2-\gamma}\left(\frac{32\overline{\zeta}^{2+2\alpha} \mathbb{E}[|\epsilon_1|^{2+2\alpha}]+32 \mathbb{E}[|\eta_1|^{2+2\alpha}]}{n^\alpha}\right)^{\frac{1}{1+\alpha}}.\label{eq9}
	\end{equation}                         
	By letting the right term of \eqref{eq9} equal to $p$, we obtain the statistical test $\Psi$.
\end{proof}

\begin{remark}\label{remark4.1}
	In some traditional detection problems with i.i.d. normal inputs, the outputs $\{Y_i\}_{i=1}^n$ also form an i.i.d. normal sequence. A common method is to use CLT to approximate $\sum_{i=1}^{n}Y_i^2$ as a random variable that follows a normal distribution for the subsequent estimation of error probabilities in hypothesis testing. However, the error introduced by this distribution approximation is not considered. Neglecting the error from the approximation may lead to an actual collected sample size that is smaller than the theoretically required sample size needed to achieve a given significance level. In other words, the actual sample size collected may therefore be insufficient to attain the specified significance level.
\end{remark}  

Recall that for each fixed $P\in\mathcal{P}$, the logarithmic moment generating function of $\eta_1^2$ under $E_P$ is defined by
\begin{equation}
\Lambda_P(\lambda):=\log E_P[e^{\lambda\eta_1^2}],\ \lambda\in \mathbb{R}.\notag
\end{equation} 
The Fenchel-Legendre transform of $\Lambda_P$ is the function defined by 
\begin{equation}
\Lambda_P^*(x):=\sup_{\lambda\in\mathbb{R}}\{\lambda x-\Lambda_P(\lambda)\},\ x\in \mathbb{R}.\notag
\end{equation}
The following theorem gives an asymptotic estimation of the upper probability of false alarms with respect to the number of samples $n$.
\begin{theorem}\label{th4.2}
	Under the setting of Theorem \ref{th2}, if $\underline{\sigma}^2\underline{\zeta}^2>\overline{\nu}^2$,  
	\begin{equation}
	\overline{C}:=\sup_{P\in\mathcal{P}}\left\lbrace  -\Lambda_P^*\left( \frac{\underline{\sigma}^2\underline{\zeta}^2+\overline{\nu}^2}{2}\right) \right\rbrace <0,\notag
	\end{equation}
	and $\Lambda_P(\lambda_P)<\infty$ for some $\lambda_P>0$ under each $P\in\mathcal{P}$,
	then 
	\begin{equation}
	\sup_{P\in\mathcal{P}}P\left(\sum_{i=1}^n\eta_i^2> n \gamma_n \right)=O\left( e^{\overline{C}n}\right) ,\label{eq16}
	\end{equation}
	which means the upper probability of false alarms of the statistical test $\Psi$ decays to zero exponentially as $n\rightarrow\infty$.
\end{theorem}
\begin{proof}
	The upper probability of false alarms is
	\begin{equation}
	\sup_{P\in\mathcal{P}}\sup_{X\in\Sigma[\underline{\zeta},\overline{\zeta}]}P\left(\sum_{i=1}^{n}Y_i^2>n \gamma_n\right)_{H_1}=\sup_{P\in\mathcal{P}}P\left(\frac{\sum_{i=1}^{n}\eta_i^2}{n}>\gamma_n\right).\notag
	\end{equation}
	By the definition of $\gamma_n$ in \eqref{eq19}, we know that there exist a $N_1\in\mathbb{N}$ such that $\gamma_n>\frac{\underline{\sigma}^2\underline{\zeta}^2+\overline{\nu}^2}{2}$ for $n\geq N_1$.
	Then, for any $P\in\mathcal{P}$, there is
	\begin{equation}
	P\left(\frac{\sum_{i=1}^{n}\eta_i^2}{n}>\gamma_n\right)\leq P\left(\frac{\sum_{i=1}^{n}\eta_i^2}{n}\geq\frac{\underline{\sigma}^2\underline{\zeta}^2+\overline{\nu}^2}{2}\right),\ n\geq N_1.\notag
	\end{equation}
	For each fixed $P\in\mathcal{P}$, due to $\frac{\underline{\sigma}^2\underline{\zeta}^2+\overline{\nu}^2}{2}>\overline{\nu}^2\geq E_p[\eta_1^2]$, then by Lemma \ref{le4}, we know that
	\begin{equation}
	P\left(\frac{\sum_{i=1}^{n}\eta_i^2}{n}\geq\frac{\underline{\sigma}^2\underline{\zeta}^2+\overline{\nu}^2}{2}\right)\leq e^{-n\Lambda_P^*\left( \frac{\underline{\sigma}^2\underline{\zeta}^2+\overline{\nu}^2}{2}\right) }\leq e^{\overline{C}n}.\notag
	\end{equation}
	Hence, we have
	\begin{equation}
	\sup_{P\in\mathcal{P}} P\left(\frac{\sum_{i=1}^{n}\eta_i^2}{n}>\gamma_n\right)
	\leq 
	e^{\overline{C}n},\ n\geq N_1,\notag
	\end{equation}
	which means \eqref{eq16} holds.
\end{proof}

\subsection{An Example}\label{sec4.3}

The following is an example when the noise follows a normal distribution, but its variance varies due to the uncertainty of probability measures.
\begin{example}\label{example4.1}
	Asuume the noise $\eta_i\sim N(0,\sigma_P^2)$ for each $i=1,2,...$ under each $P\in\mathcal{P}$, and
	\begin{equation}
	\mathbb{E}[\eta_1^2]=\sup_{P\in\mathcal{P}}\sigma_P^2=\overline{\nu}^2,\  -\mathbb{E}[-\eta_1^2]=\inf_{P\in\mathcal{P}}\sigma_P^2=\underline{\nu}^2,\  \overline{\nu}^2<\underline{\sigma}^2\underline{\zeta}^2. \notag
	\end{equation} 
	By a simple calculation,
	\begin{equation}
	\tilde{\Lambda}_P(\lambda):=\log E_P[e^{\lambda\eta_1^2}]=\left\{
	\begin{aligned}
	&+\infty,&&\lambda\geq\frac{1}{2\sigma_P^2},\\
	&-\frac{1}{2}\log(1-2\lambda\sigma_P^2),&&\lambda<\frac{1}{2\sigma_P^2},\notag
	\end{aligned}
	\right.
	\end{equation}
	\begin{equation}
	\tilde{\Lambda}_P^*(x):=\sup_{\lambda\in\mathbb{R}}\{\lambda x-\tilde{\Lambda}_P(\lambda)\}=\left\{
	\begin{aligned}
	&\frac{1}{2\sigma_P^2}x-\frac{1}{2}\log x+\log\sigma_{P}-\frac{1}{2},&&x>0,\\
	&+\infty,&&x\leq0.\notag
	\end{aligned}
	\right.
	\end{equation}
	It's clear that 
	\begin{equation}
	0=\tilde{\Lambda}_P^*(\sigma_P^2)=\inf_{x\in\mathbb{R}}\tilde{\Lambda}_P^*(x),\notag
	\end{equation}
	and $\tilde{\Lambda}_P^*(x)$ is nonincreasing on $(-\infty,\sigma_P^2)$ and nondecreasing on $(\sigma_P^2,\infty)$.\\
	By Lemma \ref{le4}, 
	\begin{equation}
	P\left(\frac{\sum_{i=1}^{n}\eta_i^2}{n}\geq\frac{\underline{\sigma}^2\underline{\zeta}^2+\overline{\nu}^2}{2}\right)\leq e^{-n\tilde{\Lambda}_P^*\left( \frac{\underline{\sigma}^2\underline{\zeta}^2+\overline{\nu}^2}{2}\right) }.\notag
	\end{equation}
	Note that 
	\begin{align*}
	\sup_{P\in\mathcal{P}}\left( -\tilde{\Lambda}_P^*\left( \frac{\underline{\sigma}^2\underline{\zeta}^2+\overline{\nu}^2}{2}\right) \right)
	=&-\inf_{\sigma_P^2\in[\underline{\nu}^2,\overline{\nu}^2]}\tilde{\Lambda}_P^*\left( \frac{\underline{\sigma}^2\underline{\zeta}^2+\overline{\nu}^2}{2}\right)\notag\\
	=&\frac{1}{4}+\frac{1}{2}\log\left(\frac{\underline{\sigma}^2\underline{\zeta}^2+\overline{\nu}^2}{2} \right)-\log\overline{\nu}-\frac{\underline{\sigma}^2\underline{\zeta}^2}{4\overline{\nu}^2}=:\hat{C}<0,\label{eq21}
	\end{align*}
	Therefore, 
	\begin{equation}
	\sup_{P\in\mathcal{P}} P\left(\frac{\sum_{i=1}^{n}\eta_i^2}{n}>\gamma_n\right)=O\left( e^{\hat{C}n}\right).\notag
	\end{equation}
\end{example}

\end{document}